\documentclass[letterpaper, 10 pt, conference]{ieeeconf} 
\IEEEoverridecommandlockouts	
\usepackage{amsmath,amssymb,amsfonts}
\usepackage{graphicx}
\usepackage{xcolor}

\newtheorem{assumption}{Assumption}
\newtheorem{problem}{Problem}
\newtheorem{theorem}{Theorem}

\newtheorem{remark}{Remark}
\newtheorem{lemma}{Lemma}
\newtheorem{definition}{Definition}

\newtheorem{fact}{Fact}
\newenvironment{proofof}[1]{\textit{Proof of #1.}}{\hfill \hspace*{0.01pt}  \hfill$\square$}

\newcommand*{\textop}[1]{\mathbin{\text{#1}}}
\makeatletter
\def\ifemptyarg#1{%
  \if\relax\detokenize{#1}\relax 
    \expandafter\@firstoftwo
  \else
    \expandafter\@secondoftwo
  \fi}
\makeatother
\newcommand{\real}{\mathbb{R}}
\newcommand{\smat}[1]{\ensuremath{\left[\begin{smallmatrix}#1\end{smallmatrix}\right]}}
\newcommand{\bmat}[1]{\ensuremath{\begin{bmatrix}#1\end{bmatrix}}}
\newcommand{\td}[1][]{\ifemptyarg{#1}{\ensuremath{t^{\textop{d}}}}{\ensuremath{t^{\textop{d}}_{#1}}}}
\newcommand{\Viss}{V}
\newcommand{\barAbf}{\bar{\mathbf{A}}}
\newcommand{\barBbf}{\bar{\mathbf{B}}}
\newcommand{\barQbf}{\bar{\mathbf{Q}}}
\newcommand{\barZbf}{\bar{\boldsymbol{\zeta}}}

\allowdisplaybreaks

\title{\LARGE \bf
Data-driven input-to-state stabilization\\ with respect to measurement errors
}

\author{Hailong Chen, Andrea Bisoffi, Claudio De Persis
\thanks{H. Chen, C. De Persis are with the Engineering and Technology Institute, University of Groningen, 9747 AG, The Netherlands. Email: {\tt\small\{hailong.chen, c.de.persis\}@rug.nl}}%
\thanks{A. Bisoffi is with the Department of Electronics, Information, and Bioengineering, Politecnico di Milano, 20133, Italy. Email: {\tt\small \{andrea.bisoffi\}@polimi.it}}%
\thanks{This publication is part of the project Digital Twin with project number P18-03 of the research programme TTW Perspective which is (partly) financed by the Dutch Research Council (NWO).}
}

\begin{document}

\maketitle

\begin{abstract}
We consider noisy input/state data collected from an experiment on a polynomial input-affine nonlinear system.
Motivated by event-triggered control, we provide data-based conditions for input-to-state stability with respect to measurement errors.
Such conditions, which take into account all dynamics consistent with data, lead to the design of a feedback controller, an ISS Lyapunov function, and comparison functions ensuring ISS with respect to measurement errors.
When solved alternately for two subsets of the decision variables, these conditions become a convex sum-of-squares program.
Feasibility of the program is illustrated with a numerical example.
\end{abstract}

\section{Introduction}

Input-to-state stability (ISS) \cite{sontag1989smooth} is a fundamental property for nonlinear systems in that it guarantees bounded state for any bounded input and that the effect of initial conditions on this bound vanishes over time.
Many stability-related notions for nonlinear systems are encompassed by ISS \cite{sontag2008input}.
In particular, ISS addresses the case of so-called ``open systems'' where one would like that a certain input, intended as a disturbance, has a limited effect on the system state thanks to the ability of the system to ``dissipate'' the energy introduced by such an input.
In this sense, ISS is linked to properties like dissipativity \cite{willems1976mechanisms}, passivity \cite{ortega1997passivity}, and $L_2$-gain \cite{van1992l2}.
Indeed, \cite{sontag1995characterizations} shows that a system is ISS if and only if an ISS Lyapunov function exists and such a Lyapunov function needs to satisfy a dissipativity-like inequality \cite[\S 3.3]{sontag2008input}.

Event-triggered control constitutes an example of the utility of the ISS property since a vast portion of event-triggered schemes (notably, \cite{tabuada2007event}, \cite{mazo2010iss}, \cite{postoyan2014framework}) assume (i)~ISS with respect to the measurement error for the closed loop system, so that the measurement error induced by event-based sampling can be tolerated, and (ii)~the existence of an ISS Lyapunov function with respect to the measurement error, which plays a key role in the construction of the event-triggering strategy. This is further elaborated in Section~\ref{sec:prel:event_trig}.

This work deals with enforcing ISS with respect to the measurement error, motivated by the relevance of this property in the context of event-triggered control \cite{scheres2022robustifying}.
Departing from model-based approaches, we would like to design a feedback controller that enforces ISS with respect to the measurement error based on noisy input\slash state data collected in an open-loop experiment, for the class of nonlinear systems given by input-affine polynomial systems.
In a nutshell, our approach consists of (i)~characterizing, through the collected noisy data, the set of all polynomial dynamics that could have generated such data, along the lines of set-membership identification \cite{milanese2004set}; (ii)~finding data-based conditions under which a controller, an ISS Lyapunov function and suitable comparison functions can be designed for all such dynamics; (iii) carrying out the design in an automated fashion, thanks to tools for sum-of-squares (SOS) optimization \cite{parrilo2003semidefinite,jarvis2005control,chesi2010lmi}. 
Since noisy data prevent exact identification of the actual system, our conditions offer robustness to uncertainty.
Along with robustness guarantees and automated optimization-based design of controller and ISS certificate, appealing features of our result are that it requires only high-level prior knowledge on the structure of the system and on the upper bound of the magnitude of the noise, and is a stepping stone towards event-triggered control.
These aspects are the contribution of this work.
The limitations of the result are that (i)~the obtained conditions are bilinear in the decision variables and, to leverage convexity, we need to solve for them in two alternate steps, see the end of Section~\ref{sec:main} for details; (ii)~we assume to approximately measure the time derivative of the state. 
Still, as for (i), most model-based approaches also suffer from the same bilinearity \cite{jarvis2005control,majumdar2013control} and, if the controller is fixed, our conditions become convex in the ISS Lyapunov function and the comparison functions; as for (ii), viable surrogates to the knowledge of the state derivative are discussed in Section~\ref{sec:probl_form}.

Based on the above discussion, this work is related to recent works on data-based verification\slash design for dissipativity and, tangentially, data-based event-triggered control.
Within the thread of dissipativity, \cite{koch2020verifying} verifies dissipativity and \cite{van2022data} enforces it from noisy input\slash state data, and \cite{romer2019one,rosa2023data} from input\slash output data, but all of these works consider linear systems.
Verification of dissipativity for polynomial systems has been considered in \cite{martin2021dissipativity,martin2022determining} where an analysis problem is solved, instead of control synthesis to enforce ISS of the closed-loop system.
Within data-driven event-triggered control, we mention \cite{de2022event,digge2022data,wang2023data} where linear systems are considered.
Whereas our previous works address data-based control design for asymptotic stability or invariance \cite{guo2020learning,bisoffi2022data,luppi2021data} for polynomial systems, we consider here a different problem and property, namely, enforcing ISS with respect to the measurement error, motivated by its relevance for event-triggered control.

\emph{Structure.} 
We introduce preliminaries and elaborate on the motivating application in Section~\ref{sec:prel}.
Our data-based setting and problem formulation are in Section~\ref{sec:probl_form}.
Section~\ref{sec:main} contains our main result to design a controller enforcing ISS with respect to the measurement error from data.
This design is tested on a numerical example in Section~\ref{sec:example}, where we also show its potential for data-based event-triggered control.

\section{Preliminaries}
\label{sec:prel}

\subsection{Notation}
\label{sec:prel:notation}

In the sequel we use sum-of-squares (SOS) polynomials and SOS matrix polynomials: we refer to, e.g., \cite{parrilo2003semidefinite,jarvis2005control,chesi2010lmi}, for excellent surveys on SOS.
The identity matrix is $I$.
We write $A\succ 0$ if a symmetric matrix $A$ is positive definite and $A \succ B$ if $A-B\succ 0$.
The set of nonnegative integer numbers is $\mathbb{N}$ and $\mathbb{N}_{\geq i}:= \mathbb{N} \cap \{j \colon j \ge i \} $.
The $n$-dimensional Euclidean space is $\mathbb{R}^n$.
The Euclidean norm of a vector $x\in\mathbb{R}^n$ is $|x|$.
A function $\alpha:\mathbb{R}_{\geq0}\rightarrow\mathbb{R}_{\geq0}$ is a class $\mathcal{K}_\infty$ function if it is continuous, strictly increasing, unbounded and satisfies $\alpha(0)=0$.
A function $\beta:\mathbb{R}_{\geq0} \times \mathbb{R}_{\geq0} \rightarrow\mathbb{R}_{\geq0}$ is a class $\mathcal{KL}$ function if, for each $t \ge 0$, $\beta(\cdot,t)$ is a class $\mathcal{K}_\infty$ function and, for each $r \ge 0$, $\beta(r,\cdot)$ is decreasing and $\lim_{t \to +\infty} \beta(r,t) = 0$ \cite[p.~166]{sontag2008input}.
For a matrix $A$, $\| A \|$ is its induced 2-norm.
We abbreviate the symmetric matrix $\left[\begin{smallmatrix}A&B^\top\\B&C\end{smallmatrix}\right]$ to $\left[\begin{smallmatrix}A&\star\\B&C\end{smallmatrix}\right]$.

\subsection{Auxiliary results for the sequel}
\label{sec:prel:aux_res}

Lemma~\ref{lemma:Kinf} will allow us to design polynomial class $\mathcal{K}_\infty$ functions and Lemma~\ref{lemma:uncert_remov} to bound the data-based uncertainty.

\begin{lemma}
\label{lemma:Kinf}
Consider $\alpha \colon \real_{\ge 0} \to \real_{\ge 0}$ defined as $\alpha(r):=\sum_{k=1}^{N}c_{k}r^{2k}$.
The function $\alpha$ is class $\mathcal{K}_{\infty}$ if the scalars $c_1,c_2,\dots,c_{N}$ satisfy $c_1\geq0,c_2\geq0,\dots,c_{N}\geq0$ and $c_1+c_2+\dots+c_N>0$.
\end{lemma}
\begin{proof}
The function $\alpha$ is continuous and $\alpha(0)=0$.
It is strictly increasing if $\alpha(r_2)>\alpha(r_1)$ for all $r_1$ and $r_2$ with $r_2>r_1\geq0$.
For two arbitrary $r_2$ and $r_1$ such that $r_2>r_1\geq0$, $\alpha(r_2)-\alpha(r_1)=\sum_{k=1}^{N}c_k(r_2^{2k}-r_1^{2k})$.
This sum of nonnegative quantities is zero only if $c_1=c_2=\dots=c_N=0$. This is excluded by hypothesis and, thus, $\alpha(r_2)>\alpha(r_1)$.
Since at least one of $c_1$, \dots, $c_N$ is nonzero as a consequence of the hypothesis, $\alpha$ grows unbounded.
\end{proof}

\begin{lemma}
\label{lemma:uncert_remov}
Given matrices $C=C^\top\in\mathbb{R}^{p\times p}$, $E\in\mathbb{R}^{p\times m}$, $G\in\mathbb{R}^{n\times p}$ and $\bar{F}=\bar{F}^\top \succeq 0$, the fact that
\begin{equation*}
\begin{aligned}
C+EFG+G^\top F^\top E^\top\preceq0 \quad \forall F \colon F^\top F\preceq\bar{F}
\end{aligned}
\end{equation*}
is implied by the existence of a scalar $\lambda>0$ such that
\begin{equation*}
\begin{aligned}
C+\lambda EE^\top+ \lambda^{-1} \,G^\top \bar{F} G\preceq0.
\end{aligned}
\end{equation*}
\end{lemma}
\begin{proof}
The claim follows by completing the square as
\begin{align*}
& 0 \succeq C+\lambda EE^\top+ \lambda^{-1} \,G^\top \bar{F} G \succeq 
C +EFG+G^\top F^\top E^\top\\
& \qquad + (\sqrt{\lambda} E - \sqrt{\lambda}^{-1} G^\top F^\top) (\sqrt{\lambda} E^\top - \sqrt{\lambda}^{-1} F G ) 
\end{align*}
by $\lambda > 0$ and $F^\top F\preceq\bar{F}$.
\end{proof}

\begin{figure}
\begin{centering}
\includegraphics[width=0.7\linewidth]{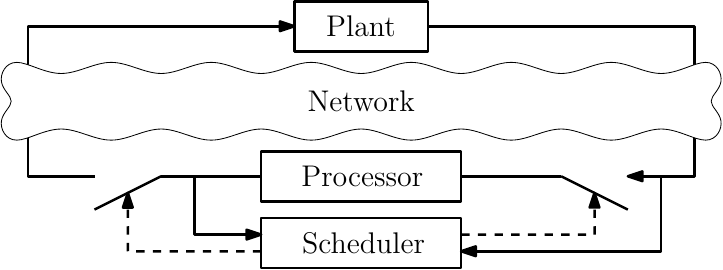}
\par\end{centering}
\caption{Event-triggered control scheme.}\label{fig1}
\end{figure}

\subsection{Motivation for ISS with respect to the measurement error based on event-triggered control}
\label{sec:prel:event_trig}

In this section, we recall the basic ideas of model-based event-triggered control \cite{tabuada2007event} to motivate our study on enforcing ISS with respect to  measurement error, from data.

Event-triggered control considers the scheme in Fig.~\ref{fig1}, where a network connects plant and controller.
The plant is described, for state $x \in \real^n$ and input $u \in \real^m$, by
\begin{equation}
\label{plant_generic}
\dot x = F(x,u).
\end{equation}
The processor computes the control action $u=k(x)$ when prompted by the scheduler.
To save communication and computation resources, the state is sampled and communicated via the network \emph{only} when necessary.
This is decided based on a stability condition involving ISS with respect to the measurement error, as explained below.
Suppose for simplicity that as soon as a new measurement is collected at times
$t_0:= 0, t_1, t_2, \dots,t_i, t_{i+1}\dots$, the controller update is computed and applied.
Between any two consecutive controller updates, the zero-order-hold policy
\begin{equation}\label{contrl}
u(t)=k(x(t_i)),\quad t\in\left[t_i,t_{i+1}\right),i\in\mathbb{N}_{\ge 0}
\end{equation}
is adopted and the measurement error is
\begin{equation*}
e(t):=x(t_i)-x(t),\quad t\in[t_i,t_{i+1}),i\in\mathbb{N}_{\ge 0},
\end{equation*}
so that $e(t_i)=0$ for all $i\in\mathbb{N}_{\ge 0}$.
With the introduction of $e$, the control action in~\eqref{contrl} is reformulated as
\begin{equation*}
u(t)=k\big(x(t)+e(t)\big),\quad t\in\left[t_i,t_{i+1}\right),i\in\mathbb{N}_{\ge 0}
\end{equation*}
and the design of the event-triggered scheme can be carried out by considering the closed-loop system
 \begin{equation}\label{sys4_generic}
\dot{x}= F(x,k(x+e)).
\end{equation}
Indeed, to determine when to sample and to analyze stability properties, one aims at enforcing that \eqref{sys4_generic} is ISS with respect to the measurement error $e$, as per the next definition where $\| e \|_\infty$ is the (essential) supremum norm of $e(\cdot)$.
\begin{definition}[{\cite[p.~171]{sontag2008input}}]
\label{def:ISS}
System~\eqref{sys4_generic} is ISS with respect to measurement error $e$ if there exist a class $\mathcal{KL}$ function $\beta$ and a class $\mathcal{K}_\infty$ function $\gamma$ such that the bound
\begin{align*}
|x(t)| \le \beta(|x_0|, t) + \gamma(\| e \|_\infty )
\end{align*}
holds for all solutions (i.e., for all signals $e(\cdot)$, all initial conditions $x_0$ and all $t \ge 0$).
\end{definition}

A standard way to guarantee ISS is through an ISS Lyapunov function, defined next, since a system is ISS (with respect to some input) if and only if it admits a smooth ISS Lyapunov function \cite{sontag1995characterizations}.
\begin{definition}[{\cite[p.~178]{sontag2008input}}]
\label{def:ISS_Lyapunov_function}
A smooth function $\Viss:\mathbb{R}^n\rightarrow\mathbb{R}_{\geq0}$ is an ISS Lyapunov function for the closed-loop system \eqref{sys4_generic} if there exist class $\mathcal{K}_\infty$ functions $\alpha_1,\alpha_2,\alpha_3$ and $\alpha_4$ satisfying, for all $x$ and $e$,
\begin{subequations}
\label{propr_ISS_Lyap_fun}
\begin{align}
& \alpha_1(| x|)\leq \Viss(x)\leq\alpha_2(| x|)\label{def1-1} \\
& \langle\nabla{\Viss}(x),F(x,k(x+e))\rangle\leq-\alpha_3(| x|)+\alpha_4(| e|).\label{def1-2}
\end{align}
\end{subequations}
\end{definition}
\medskip

Suppose now that sampling and controller updates occur as soon as, for $\sigma \in (0,1)$, the condition
\begin{align}
\label{event_trigg_cond}
\alpha_4(|e|) \le \sigma \alpha_3(| x|)
\end{align}
is violated, where $\alpha_3$ and $\alpha_4$ are from Definition~\ref{def:ISS_Lyapunov_function}.
Then,
\begin{equation}
\label{implied_strict_decrease}
\langle\nabla{\Viss}(x), \, F(x,k(x+e)) \rangle\leq(\sigma-1)\alpha_3(|x|)
\end{equation}
and strict decrease of $\Viss$ for $x \neq 0$ holds. This, along with~\eqref{def1-1}, ensures global asymptotic stability of the origin.

In view of this reasoning, the existence of an ISS Lyapunov function with the properties in Definition~\ref{def:ISS_Lyapunov_function} is the cornerstone of most event-triggered control schemes.
Nonetheless, the existence of such an ISS Lyapunov function is typically assumed as a starting point in the event-triggered literature \cite{tabuada2007event}, \cite{mazo2010iss}, \cite{postoyan2014framework}.
Here, we would like to show that, for polynomial input-affine nonlinear systems, such an ISS Lyapunov function can be computed from data.

\section{Data-based problem formulation}
\label{sec:probl_form}

As a class of the nonlinear systems in~\eqref{plant_generic}, we consider input-affine nonlinear systems with polynomial dynamics
\begin{equation}\label{sys1}
\dot{x}=f_{\star}(x)+g_{\star}(x)u
\end{equation}
where $f_{\star}$ and $g_{\star}$ have polynomials as their elements.
We consider this relevant class of~\eqref{plant_generic} because, in general, polynomial vector fields can approximate smooth vector fields tightly (on compact sets) and, here, they enable combining data-based conditions with powerful design tools provided by SOS.
The actual expressions of $f_\star$ and $g_\star$ are unknown to us; nonetheless, we make the next assumption on them.
\begin{assumption}
\label{Assum1}
We know $Z(x) \in \real^N$ and $W(x) \in \real^{M \times m}$ that have monomials of $x$ as their elements and are such that there exist constant coefficient matrices $A_\star \in \real^{n \times N}$ and $B_\star \in \real^{n \times M}$ satisfying, for each $x \in \real^n$,
\begin{align*}
f_\star(x) = A_\star Z(x), \quad g_\star(x) = B_\star W(x).
\end{align*}
Moreover, $Z(0)=0$.
\end{assumption}
\medskip
Thanks to Assumption~\ref{Assum1}, \eqref{sys1} can be rewritten as
\begin{equation}\label{sys2}
\dot{x}=A_{\star}Z(x)+B_{\star}W(x)u
\end{equation}
for some \emph{known} $Z$ and $W$, which amount to lists of monomials, and \emph{unknown} coefficient matrices $A_{\star}$ and $B_{\star}$.
We make the next remark on Assumption~\ref{Assum1} and the selection of regressors $Z$ and $W$.
\begin{remark}
\label{remark:choice_Z_W}
In principle, a foolproof way of satisfying Assumption~\ref{Assum1} is to let $Z$ and $W$ contain all monomials with degree less than a large positive integer since Assumption~\ref{Assum1} only requires that if a monomial is present in $f_\star$ (or $g_\star$) it must be present in $Z$ (or $W$), but not vice versa.
This choice, however, can be overly conservative because the redundant terms can lead to an increased computational cost and possible infeasibility of the resulting SOS program (in our case, the program \eqref{dataDrivenProgram} below).
Still, even from high-level knowledge on the system under study \cite{ahmadi2021learning}, it may be apparent which monomials would appear in $f_\star$ (or $g_\star$), and thus in $Z$ (or $W$); otherwise, techniques such as those in \cite{brunton2016discovering} can be preliminarily employed.
In this regard, a parsimonious choice of the monomials in $Z$ and $W$ is best suited.
\end{remark}

To compensate for the lack of knowledge of $A_{\star}$ and $B_{\star}$, our approach is to gather information about the system from data and, based only on such data, design a state-feedback controller enforcing ISS with respect to the measurement error.
Data are collected in an open-loop experiment.
We consider $T$ data that are generated by~\eqref{sys2} in the inevitable presence of a noise term $d$, namely, for $i=0,\dots,T-1$,
\begin{align}\label{sys5}
\dot{x}(\td[i])=A_{\star}Z(x(\td[i]))+B_{\star}W(x(\td[i]))u(\td[i])+d(\td[i])
\end{align}
where we measure input, state and state derivative at time $\td[i]$ (the superscript $^{\textop{d}}$ denotes ``data''), but we know only a norm bound on the noise samples, as in the next assumption.
\begin{assumption}\label{Assum2}
For $\delta>0$, $d(\td[0])$, \dots, $d(\td[T-1])$ belong to
\begin{equation}
\label{setD}
\mathcal{D}:=\left\{d\in\mathbb{R}^{n}\colon | d|^2\leq \delta\right\}.
\end{equation}
\end{assumption}
\medskip

For simplicity, we assume to measure the state derivative $\dot{x}$ at times $\td[0]$, \dots, $\td[T-1]$.
When not available, $\dot{x}$ can be recovered from a denser sampling of $x$, e.g., using techniques from continuous-time system identification \cite{garnier2003continuous}: at any rate, these techniques allow reconstructing $\dot{x}$ with some error which we account for through noise $d$.
See \cite[Appendix A]{de2022event} for an alternative approach that does not require measuring $\dot{x}$ (but leads anyhow to the sets $\mathcal{I}_i$ and $\mathcal{I}$ defined below in~\eqref{def11} and \eqref{set_I}).
Moreover, we emphasize that sampling is \emph{not} required to be uniform (i.e., we do not need $\td[T-1]- \td[T-2] = \dots = \td[2]- \td[1] = \td[1]- \td[0]$) and, as a matter of fact, data points $\{ u(\td[i])$, $x(\td[i])$, $\dot{x}(\td[i]) \}_{i=0}^{T-1}$ need \emph{not} be collected from a single trajectory but can arise from multiple trajectories.

Based on the collected data, we can characterize  the set of matrices $[A\ B]$ consistent with data point $\{ u(\td[i])$, $x(\td[i])$, $\dot{x}(\td[i]) \}$ and the instantaneous bound $\mathcal{D}$ in~\eqref{setD} as
\begin{equation}\label{def11}
\mathcal{I}_i:=\big\{[A~B] \colon\dot{x}(\td[i])=
[A~B]\smat{Z(x(\td[i]))\\W(x(\td[i]))u(\td[i])} +d,|d|^2\leq \delta\big\},
\end{equation}
namely, the set of all matrices $[A\ B]$ that could have generated data point $\{ u(\td[i])$, $x(\td[i])$, $\dot{x}(\td[i]) \}$ for some $d$ complying with bound $\mathcal{D}$, cf.~\eqref{sys5}.
The set of matrices consistent with all data points and the instantaneous bound $\mathcal{D}$ is then
\begin{align}\label{set_I}
\mathcal{I}:=\bigcap_{i=0}^{T-1}\mathcal{I}_i. 
\end{align}
The presence of the noise $d$ during the experiment prevents from identifying the system unambigously, as is the case in set-membership system identification \cite{milanese2004set}, and requires us to design a feedback controller enforcing ISS with respect to the measurement error for all matrices $[A\ B] \in \mathcal{I}$, all of which could have generated the data.
Based on Section~\ref{sec:prel:event_trig}, this amounts to the design of a feedback controller $k$ and, as per Definition~\ref{def:ISS_Lyapunov_function}, an ISS Lyapunov function $\Viss$ and class $\mathcal{K}_\infty$ functions $\alpha_1$, $\alpha_2$, $\alpha_3$, $\alpha_4$ so that the conditions in~\eqref{propr_ISS_Lyap_fun} hold true, based on our priors and the available data. 
This results in the next problem statement.

\begin{problem}
\label{probl_form}
For data points $\{ u(\td[i])$, $x(\td[i])$, $\dot{x}(\td[i]) \}_{i=0}^{T-1}$ and the resulting set in $\mathcal{I}$ in~\eqref{set_I}, design a controller $u=k(x)$, an ISS Lyapunov function $\Viss$ and class $\mathcal{K}_\infty$ functions $\alpha_1$, $\alpha_2$, $\alpha_3$, $\alpha_4$ such that for all $[A \ B] \in \mathcal{I}$ and all $(x, e)$
\begin{subequations}
\begin{align}
& \hspace*{-5pt}\alpha_1(| x|)\leq \Viss(x)\leq\alpha_2(| x|) \label{sandwichVISS_pol} \\
& \hspace*{-5pt}\langle\nabla{\Viss}(x), \bmat{A & B} \! \smat{Z(x)\\ W(x) k(x+e)}\rangle \le \! -\alpha_3(| x|)+\alpha_4(| e|).\!\! \label{decreaseVISS_pol}
\end{align}
\end{subequations}
\end{problem}
\smallskip
We emphasize that this problem involves finding a solution $V$ to a dissipativity-like inequality as that in~\eqref{decreaseVISS_pol}, from data.

Before solving Problem~\ref{probl_form} in Section~\ref{sec:main}, we conclude this section by introducing an ellipsoidal overapproximation of the set $\mathcal{I}$, which relies on~\cite{bisoffi2021trade}, \cite{luppi2021data} and extends the approach in \cite[\S 3.7.2]{boyd1994linear}.
For data point $i=0,\dots,T-1$, define
\begin{align}
&\pmb{C}_i:=\dot{x}(\td[i])\dot{x}(\td[i])^\top-\delta I,\,\pmb{B}_i:=-\begin{bmatrix}Z(x(\td[i]))\\W(x(\td[i]))u(\td[i])\end{bmatrix}\dot{x}(\td[i])^\top\!, \notag \\
&\pmb{A}_i:=\begin{bmatrix}Z(x(\td[i]))\\W(x(\td[i]))u(\td[i])\end{bmatrix}\begin{bmatrix}Z(x(\td[i]))\\W(x(\td[i]))u(\td[i])\end{bmatrix}^\top.\label{datapoint}
\end{align}
Consider the set $\bar{\mathcal{I}}$
\begin{equation*}
\begin{aligned}
\bar{\mathcal{I}} \! := \! \{[A~B]=\zeta^\top \! \colon\bar{\mathbf{B}}^\top\! \barAbf^{-1}\barBbf+\barBbf^\top\!\zeta+\zeta^\top\!\bar{\mathbf{B}}+\zeta^\top\!\barAbf \zeta \preceq I\}
\end{aligned}
\end{equation*}
where the matrices $\barAbf$ and $\barBbf$ are designed by solving
\begin{equation}\label{overapp}
\begin{aligned}
& \text{min.} & & \hspace*{-6.2pt}-\log\det \barAbf\quad(\text{over}\,\barAbf,\barBbf,\tau_0,\dots,\tau_{T-1})\\
& \text{s.t.} & & \hspace*{-6.2pt} \barAbf\succ0,\tau_i\geq0,\,\text{for}\, i=0,1,\dots,T-1\\
& & & \hspace*{-6pt}\begin{bmatrix}-I-\sum_{i=0}^{T-1}\tau_i \pmb{C}_i&\star&\star\\ \barBbf-\sum^{T-1}_{i=0}\tau_i\pmb{B}_i&\barAbf-\sum_{i=0}^{T-1}\tau_i\pmb{A}_i&\star\\\barBbf&0&-\barAbf\end{bmatrix}\preceq 0.
\end{aligned}
\end{equation}
The objective function of~\eqref{overapp} corresponds to the size of the set $\bar{\mathcal{I}}$ \cite[\S 2.2]{bisoffi2021trade} and its constraints ensure the next result.
\begin{fact}[{\cite[\S 5.1]{bisoffi2021trade}}]
\label{fact:cal_I_subset_bar_cal_I}
If $\barAbf$ and $\barBbf$ are a solution to~\eqref{overapp}, the set $\bar{\mathcal{I}}$ satisfies $\mathcal{I} \subseteq \bar{\mathcal{I}}$.
\end{fact}
The so-obtained $\bar{\mathcal{I}}$ is thus an overapproximation of $\mathcal{I}$ and we use it since, unlike the set $\mathcal{I}$, it is a matrix ellipsoid and this enables the developments of the sequel.
We would like to reassure the reader on the feasibility of the optimization program in~\eqref{overapp} by recalling the next fact, with definitions
\begin{align*}
W_0:=&
\bmat{W(x(\td[0]))u(\td[0]) & \!\!\!\dots\!\!\! & W(x(\td[T-1]))u(\td[T-1])}\in\mathbb{R}^{M\times T}\\
Z_0:=&
\bmat{
Z(x(\td[0]))\hspace*{8.7mm} & \!\!\!\dots\!\!\! & Z(x(\td[T-1]))
}\in\mathbb{R}^{N\times T}.
\end{align*}
\begin{fact}[{\cite[Lemma 2]{luppi2021data}}]
\label{fact:suff_cond_feas}
If the matrix $\smat{Z_0\\W_0}$ has full row rank, then the optimization program in \eqref{overapp} is feasible.
\end{fact}
Intuitively, collecting more data points can ensure that $\smat{Z_0\\W_0}$ becomes full row rank if it is not, since these additional data points constitute additional columns of $\smat{Z_0\\W_0}$; hence, Fact~\ref{fact:suff_cond_feas} suggests that the more data points, the higher the chance that $\smat{Z_0\\W_0}$ has full row rank and, in turn, that \eqref{overapp} is feasible.

If $\barAbf$ and $\barBbf$ are a solution to~\eqref{overapp}, we have $\barAbf\succ 0$ by construction. Then, we can define 
\begin{align}\label{zetaBoldQbold}
\barZbf :=-\barAbf^{-1}\barBbf \text{ and }\barQbf :=I,
\end{align}
and rewrite the set $\mathcal{\bar{I}}$ as
\begin{equation}\label{data1}
\mathcal{\bar{I}}=\{[A~B]=\zeta^\top=(\barZbf+\barAbf^{-\frac{1}{2}}\Upsilon \barQbf^{\frac{1}{2}})^\top\colon \|\Upsilon\|\leq 1 \}.
\end{equation}

\section{Main result: learning an ISS Lyapunov Function from data}
\label{sec:main}

In this section, we give our main result.
With $\barAbf$, $\barZbf$ and $\barQbf$ readily obtained from noisy data via~\eqref{overapp}, we provide conditions to find a controller, an ISS Lyapunov function, and comparison functions for all possible dynamics consistent with data.
The next theorem contains these conditions.

\begin{theorem}[Data-driven noisy ISS]\label{thm:ISSlf}
For data points $\{ u(\td[i])$, $x(\td[i])$, $\dot{x}(\td[i]) \}_{i=0}^{T-1}$ and under Assumptions~\ref{Assum1}-\ref{Assum2}, let the optimization program in \eqref{overapp} be feasible.
Suppose there exists a scalar $\varepsilon > 0$, a vector polynomial $k$ with $k(0)=0$, polynomials $\alpha_1$, \dots, $\alpha_4$, polynomials $\Viss$ and $\lambda$, SOS polynomials $s_{1},s_{2}$, $s_{3}$, an SOS matrix polynomial $s_{4}$ such that \eqref{dataDrivenProgram}, displayed over two columns, holds for all $r$, $x$, $e$.
Then,
\begin{align}
\label{closed_loop_data_driven}
\dot{x} = \bmat{A & B} \bmat{ Z(x) \\ W(x) k(x+e) }
\end{align}
is ISS with respect to the measurement error $e$ for all $[A\ B] \in \bar{\mathcal{I}}$, and in particular for $[A_\star \ B_\star]$.
\begin{figure*}
\begin{subequations}\label{dataDrivenProgram}
\begin{align}
& \alpha_i(r) = \sum_{k=1}^{N_i}c_{ik}r^{2k}, \, c_{ik} \ge 0, \, \sum_{k=1}^{N_i}c_{ik} \ge \varepsilon \text{ for } i =1, \dots, 4, \; k=1, \dots, N_i \label{classKinfFunctions} \\
& \Viss(x)-\alpha_1(| x|) = s_{1}(x),\quad \alpha_2(| x|)-\Viss(x) =s_{2}(x), \quad \lambda(x,e)-\varepsilon =s_{3}(x,e),\label{sandwichBound_multiplier}\\
&
\bmat{
\alpha_3(| x|)-\alpha_4(| e|)+\frac{\partial{\Viss}}{\partial{x}}(x)\barZbf^\top\smat{Z(x)\\W(x)k(x+e)}&\star&\star\\\barQbf^\frac{1}{2}\frac{\partial{\Viss}}{\partial{x}}(x)^\top&-2\lambda(x,e) I&\star\\ \lambda(x,e) \barAbf^{-\frac{1}{2}}\smat{Z(x)\\W(x)k(x+e)} &0&-2\lambda(x,e) I
} =-s_{4}(x,e),\label{decreaseISSlyapunovFunction}
\end{align}
\end{subequations}
\hrule
\end{figure*}
\end{theorem}

Before proving this result, some comments on it are appropriate. 
Condition~\eqref{classKinfFunctions} ensures that $\alpha_1$, $\alpha_2$, $\alpha_3$, $\alpha_4$ are class $\mathcal{K}_\infty$ functions by Lemma~\ref{lemma:Kinf}.
The reason for considering even powers of $r$ in them is that such powers cancel the square roots in $\sqrt{\sum_{i=1}^n x_i^2}=|x|$ or $\sqrt{\sum_{i=1}^n e_i^2}=|e|$ so that $\alpha_1(|x|)$, $\alpha_2(|x|)$, $\alpha_3(|x|)$, $\alpha_4(|e|)$ in~\eqref{sandwichBound_multiplier} and \eqref{decreaseISSlyapunovFunction} are polynomials in the components of $x$ and $e$ and, thus, SOS tools can be fruitfully applied to solve \eqref{dataDrivenProgram}.
By allowing some coefficients of functions $\alpha_i$, $i=1,\dots,4$ to be zero, we let program \eqref{dataDrivenProgram} design the maximum degree of the $\alpha_i$'s, which is only required to be nongreater than a positive integer $N_i$ of our choice.
Condition~\eqref{sandwichBound_multiplier} entails lower and upper bounds on the ISS Lyapunov function $V$, as in~\eqref{sandwichVISS_pol}, and positivity of state- and error-dependent multiplier $\lambda$, as needed in Lemma~\ref{lemma:uncert_remov}.
Finally, condition~\eqref{decreaseISSlyapunovFunction} corresponds to guaranteeing the dissipativity-like inequality \eqref{decreaseVISS_pol} for all polynomial dynamics that are consistent with data, and uses the quantities $\barZbf$, $\barQbf$, $\barAbf$ that were obtained from $\{ u(\td[i])$, $x(\td[i])$, $\dot{x}(\td[i]) \}_{i=0}^{T-1}$ by solving~\eqref{overapp}.
Theorem~\ref{thm:ISSlf} effectively solves Problem~\ref{probl_form} since $\mathcal{I} \subseteq \bar{\mathcal{I}}$ by Fact~\ref{fact:cal_I_subset_bar_cal_I}.

\begin{proofof}{Theorem~\ref{thm:ISSlf}}
By \eqref{classKinfFunctions} and Lemma~\ref{lemma:Kinf}, $\alpha_1$, \dots, $\alpha_4$ are class $\mathcal{K}_\infty$ functions (when their domain is restricted to $\real_{\ge 0}$).
Since $s_1$, $s_2$ are SOS polynomials, \eqref{sandwichBound_multiplier} implies \eqref{sandwichVISS_pol}.
If, for $\Viss$, $\alpha_3$ and $\alpha_4$, we have that for all $(x,e)$ and all $\Upsilon$ with $\|\Upsilon\| \leq 1$
\begin{equation}\label{f77}
\begin{aligned}
&\alpha_3(| x|)-\alpha_4(| e|)+\tfrac{\partial{\Viss}}{\partial{x}}(x)\barZbf^\top\smat{Z(x)\\W(x)k(x+e)}\\
&+\tfrac{\partial{\Viss}}{\partial{x}}(x)\barQbf^{\frac{1}{2}}\Upsilon^\top\mathbf{\bar{A}}^{-\frac{1}{2}}\smat{Z(x)\\W(x)k(x+e)}\leq 0,
\end{aligned}
\end{equation}
then, by~\eqref{data1}, \eqref{decreaseVISS_pol} holds for all $(x,e)$ and all $[A\ B] \in \mathcal{\bar{I}} \supseteq \mathcal{I}$, where \eqref{decreaseVISS_pol} corresponds to~\eqref{def1-2}.
In other words, by Definition~\ref{def:ISS_Lyapunov_function}, $V$ is a (smooth) Lyapunov function for~\eqref{closed_loop_data_driven} for all $[A \ B] \in \bar{\mathcal{I}}$; since $[A_\star \ B_\star] \in \mathcal{I} \subseteq \bar{\mathcal{I}}$ by Assumption~\ref{Assum2}, this would prove the statement.
We show then that by~\eqref{sandwichBound_multiplier} and \eqref{decreaseISSlyapunovFunction}, \eqref{f77} holds for all $(x,e)$ and all $\Upsilon$ with $\|\Upsilon\| \leq 1$.
We rewrite \eqref{f77} as
\begin{equation}\label{f7}
\begin{aligned}
& \alpha_3(| x|)-\alpha_4(|e|)+\tfrac{\partial{\Viss}}{\partial{x}}(x)\barZbf^\top \smat{ Z(x)\\ W(x)k(x+e)} \\
&+\tfrac{1}{2}\tfrac{\partial{\Viss}}{\partial{x}}(x)\barQbf^{\frac{1}{2}}\Upsilon^\top\mathbf{\bar{A}}^{-\frac{1}{2}} \smat{Z(x)\\W(x)k(x+e)} \\
& +\tfrac{1}{2}\left[\begin{smallmatrix}Z(x)\\W(x)k(x+e)\end{smallmatrix}\right]^\top\mathbf{\bar{A}}^{-\frac{1}{2}}\Upsilon\mathbf{\bar{Q}}^{\frac{1}{2}}\tfrac{\partial \Viss}{\partial x}(x)^\top\leq 0.
\end{aligned}
\end{equation}
Since $\varepsilon > 0$ and $s_3$ is an SOS polynomial, \eqref{sandwichBound_multiplier} implies $\lambda(x,e)>0$ for all $(x,e)\in\mathbb{R}^{n+n}$. So, by Lemma~\ref{lemma:uncert_remov}, \eqref{f7} is valid for all $(x,e)$ and all $\Upsilon$ with $\|\Upsilon\| \leq 1$ if, for all $(x,e)$,
\begin{align}
&\alpha_3(| x|)-\alpha_4(| e|)+\tfrac{\partial{\Viss}}{\partial{x}}(x)\barZbf^\top \smat{Z(x)\\W(x)k(x+e)} \notag \\
& +\tfrac{\lambda(x,e)}{2} \smat{Z(x)\\W(x)k(x+e)}^\top\mathbf{\bar{A}}^{-\frac{1}{2}}\barAbf^{-\frac{1}{2}} \smat{Z(x)\\W(x)k(x+e)} \notag \\
& +\tfrac{1}{2\lambda(x,e)}\tfrac{\partial{\Viss}}{\partial{x}}(x)\barQbf^{\frac{1}{2}}\barQbf^{\frac{1}{2}}\tfrac{\partial{\Viss}}{\partial{x}}(x)^\top \leq 0. \label{f8}
\end{align}
Applying Schur complement to \eqref{f8} yields that, for all $(x,e)$,
\begingroup
\thinmuskip=.5mu plus 1mu
\medmuskip=1.mu plus 1mu
\thickmuskip=1.5mu plus 1mu
\setlength\arraycolsep{1.pt}%
\begin{equation*}
\bmat{
\left\{
\begin{matrix}
\alpha_3(|x|)-\alpha_4(|e|)\\
+\frac{\partial{\Viss}}{\partial{x}}(x)\barZbf^\top \smat{Z(x)\\W(x)k(x+e)}
\end{matrix}
\right\}
 &\star&\star\\
\barQbf^{\frac{1}{2}}\frac{\partial{\Viss}}{\partial{x}}(x)^\top&-2 \lambda(x,e) I&\star\\
\lambda(x,e)\barAbf^{-\frac{1}{2}} \smat{Z(x)\\W(x)k(x+e)} &0&-2\lambda(x,e)I
} \preceq 0,
\end{equation*}
\endgroup%
which is implied by~\eqref{decreaseISSlyapunovFunction}.
\end{proofof}

Admittedly, the program \eqref{dataDrivenProgram}, and specifically \eqref{decreaseISSlyapunovFunction}, is bilinear, and thus nonconvex, due to the products of decision variables $V$, $k$ and $\lambda$, $k$.
To address this nonconvexity, one can use solvers tailored for bilinear programs, such as PENBMI and BMIBND.
Here, we adopt the widespread alternate approach \cite{jarvis2005control,majumdar2013control} where, in a first step, we fix $k$ and solve \eqref{dataDrivenProgram}, which becomes convex with respect to the remaining decision variables (which include $V$, $\lambda$) and, in a second step, we fix $V$, $\lambda$ and solve \eqref{dataDrivenProgram}, which becomes convex with respect to the remaining decision variables (which include $k$).
We successfully use this approach in our numerical example in Section~\ref{sec:example}.

For a given $u_{\max}(\cdot)$, bounds in the form $|k(x)| \le u_{\max}(x)$, for all $x$, are readily incorporated by adding to~\eqref{dataDrivenProgram} the convex condition $\smat{-u_{\max}(x)^2 & k(x)^\top\\ k(x) & -I} = - s_5(x)$, for all $x$, for an SOS polynomial $s_5$. 

\begin{figure}[htbp]
\centerline{\hspace*{3.5pt}\includegraphics[scale=.6]{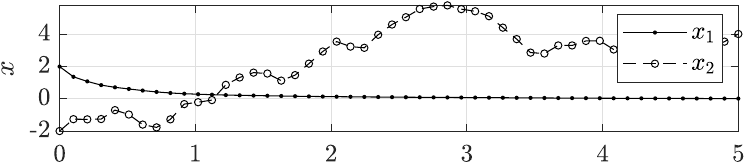}}
\centerline{\includegraphics[scale=.6]{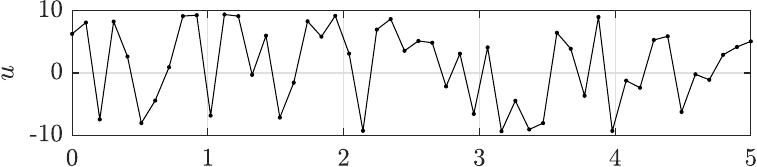}}
\centerline{\hspace*{2pt}\includegraphics[scale=.6]{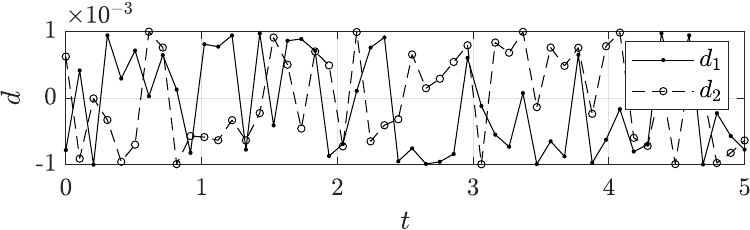}}
\caption{Data-collection experiment for the considered polynomial system.}\label{fig:data}
\end{figure}

\section{Numerical example}
\label{sec:example}

To exemplify our result for data-based design of a controller achieving ISS with respect to the measurement error, we consider the polynomial system
\begin{equation}\label{ex1}
\smat{
\dot{x}_1\\
\dot{x}_2}
=
\smat{-x_1+x_1^2x_2\\
0}
+
\smat{0\\1}u = f_{\star}(x)+g_{\star}(x)u
\end{equation}
from \cite[Example 14.11]{khalil2002nonlinear}.
Based on the considerations in Remark~\ref{remark:choice_Z_W}, we select $Z(x)=\bmat{x_1&x_1^2&x_1^2x_2&x_1x_2^2&x_2^3}^\top$ and $W(x)=1$, which satisfy Assumption~\ref{Assum1}. Note that $Z$ contains more monomials than those appearing in $f_\star$.
With these $Z$ and $W$, \eqref{ex1} yields
$A_\star =
\smat{
-1&0&1&0&0\\
0&0&0&0&0
}$ and $B_\star=
\smat{
0\\1
}$,
which are \emph{unknown} in our setting, see Section~\ref{sec:probl_form}, and used \emph{only} to generate the noisy data points on which \eqref{dataDrivenProgram} is based.
In the numerical experiment, $u$ and $d$ are taken as uniformly distributed random numbers with $|u| \le 10$ and $|d| \le 0.001$ and $x_0=(2,-2)$.
The evolutions of $x$, $u$ and the \emph{unknown} $d$ are in Fig.~\ref{fig:data}.
From these evolutions we collect data points $\{ u(\td[i])$, $x(\td[i])$, $\dot{x}(\td[i]) \}_{i=0}^{T-1}$ with $T=50$.

With these data points, we construct the matrices $\pmb{C}_i$, $\pmb{B}_i$, $\pmb{A}_i$ in \eqref{datapoint}, $i=0,\dots, T-1$, solve \eqref{overapp} by YALMIP \cite{lofberg2004yalmip} and obtain matrices $\barAbf$, $\barBbf$ and, thus, $\barZbf$ and $\barQbf$ in~\eqref{zetaBoldQbold}.

With $\barAbf$, $\barZbf$ and $\barQbf$, we solve \eqref{dataDrivenProgram} with YALMIP \cite{Lofberg2009}
imposing that the maximum degrees of $\lambda$, $k$, $V$ are respectively 4, 3, 2 at most.
As elucidated at the end of Section~\ref{sec:main}, we solve \eqref{dataDrivenProgram} alternately, first with $V$, $\lambda$ (but not $k$) among the decision variables and then with $k$ (but not $V$, $\lambda$) among the decision variables.
These two steps are repeated 3 times and require an initial guess for the controller $k$.
As an initial guess, we take $k(x)=-x_1^3-8x_2$ from \cite[Example 14.11]{khalil2002nonlinear}, which achieves global asymptotic stability when the error $e$ is zero; it is however well-known that global asymptotic stability with $e=0$ does not imply ISS with respect to $e$ in general \cite{sontag2008input}.
The solutions obtained from~\eqref{dataDrivenProgram} in this way are
\begin{align}
& k(x)= -0.0410 x_1 -1.8348 x_2 + 0.0086 x_1^2 \notag\\
& + 0.0279 x_1 x_2 + 0.0128 x_2^2 -1.3188 x_1^3 \notag\\
&-4.1114 x_1^2 x_2 -2.9410 x_1 x_2^2 -1.7240 x_2^3 \label{ctrl_sim}\\
&\Viss(x) =1.1045 x_1^2+1.3849 x_2^2 + 1.2357 x_1 x_2,\label{visssim} \\
&\alpha_1(r)=0.4808 r^2, \, \alpha_2(r) = 2.2602 r^2+ 1.3966 r^4,\notag\\
&\alpha_3(r)=0.6066 r^2+ 0.0378 r^4, \notag\\
& \alpha_4(r) = 2.8953 r^2+ 10.0643 r^4,  \, \lambda(x,e) =12.0254. \notag
\end{align}
For these polynomials, we have reported only their most significant terms and not the terms with comparatively much smaller coefficients.
Having found a certified feasible solution entails that \eqref{ex1} with feedback controller $u = k(x+e)$ for $k$ in~\eqref{ctrl_sim} (i.e., in the presence of measurement error $e$) is ISS with respect to $e$, thanks to the existence of ISS Lyapunov function $\Viss$ in~\eqref{visssim}.

To showcase this solution but without discussing the existence of a lower bound on the interevent times, we utilize the data-based designed controller $k$ and class $\mathcal{K}_\infty$ functions $\alpha_3$, $\alpha_4$ in an event-triggered control scheme, as outlined in Section~\ref{sec:prel:event_trig}.
Specifically, the control action is computed as in~\eqref{contrl} and the triggering condition is in~\eqref{event_trigg_cond}, for $\sigma=0.9$.
This results in the state evolution in Fig.~\ref{fig:et}, where convergence to zero is ensured by~\eqref{implied_strict_decrease}.
Fig.~\ref{fig:et} also shows $\alpha_{4}(| e(t)|) \le \sigma\alpha_3(| x(t)|)$ for all $t \ge 0$, thus confirming the strict decrease of the designed ISS Lyapunov function as in
\begin{align*}
\left\langle\nabla \Viss(x),
\bmat{A_\star & B_\star} \smat{ Z(x) \\ W(x) k(x+e) }\right\rangle
\leq-\alpha_3(| x|)+\alpha_4(| e|).
\end{align*}

\begin{figure}[htbp]
\centerline{\includegraphics[scale=0.62]{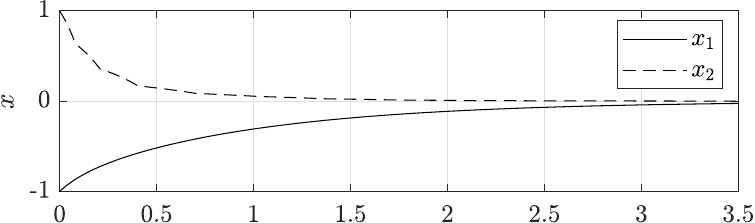}}\smallskip
\centerline{\hspace*{-2.6pt}\includegraphics[scale=0.62]{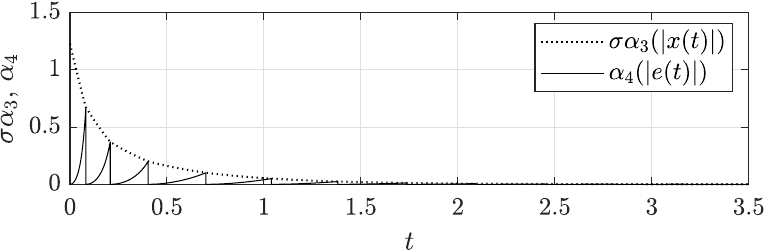}}
\caption{Evolution of the state (top) and comparison functions (bottom).}
\label{fig:et}
\end{figure}

\section{Conclusion}\label{Conclusion}

Noisy input\slash state data from an open-loop experiment on a polynomial input-affine system allow the designer to characterize all dynamics consistent with data.
We provided conditions to design, robustly for all these consistent dynamics, a feedback controller, an ISS Lyapunov function, and comparison functions that enforce ISS of the closed loop with respect to the measurement error.
The design is verified on a numerical simulation, which shows the potential of employing such controller and comparison functions for a data-based event triggered control scheme.

\bibliographystyle{IEEEtran}
\bibliography{pubs-eventtriggered}
\end{document}